\documentclass[runningheads]{llncs}

\usepackage{float}
\usepackage{stmaryrd}
\usepackage{graphicx}
\usepackage{xcolor}
\usepackage{amsmath}
\usepackage{amssymb}
\usepackage{algcompatible}
\usepackage{algorithmicx}
\usepackage[noend]{algpseudocode}
\usepackage{algorithm}
\usepackage{wrapfig}

\usepackage{alltt}

\usepackage{pifont}
\newcommand{\cmark}{\ding{51}}
\newcommand{\xmark}{\ding{55}}

\algdef{SE}[DOWHILE]{Do}{doWhile}{\algorithmicdo}[1]{\algorithmicwhile\ #1}%

\makeatletter
\newcommand{\dashedrightarrow}[1][2pt]{%
  \settowidth{\@tempdima}{$\rightarrow$}\rightarrow
  \makebox[-\@tempdima]{\hskip-1.5ex\color{white}\rule[0.5ex]{#1}{1pt}}
  \phantom{\rightarrow}
}
\makeatother

\newcommand{\integers}{\mathbb{Z}}

\newcommand{\ppre}{\ensuremath{\mathsf{Pre}}}

\newcommand{\ltrue}{\mathbf{tt}}
\newcommand{\lfalse}{\mathbf{ff}}

\newcommand{\union}{{\cup} }

\newcommand{\ourtool}{\textsc{Vajra}}
\newcommand{\freqhorn}{\textsc{FreqHorn}}
\newcommand{\viap}{\textsc{VIAP}}
\newcommand{\veriabs}{\textsc{VeriAbs}}
\newcommand{\zthree}{\textsc{Z3}}

\newcommand{\booster}{\textsc{Booster}}

\newcommand{\vaphor}{\textsc{Vaphor}}

\newcommand{\PP}{\ensuremath{\mathsf{P}}}

\newcommand{\EE}{\ensuremath{\mathsf{E}}}

\newcommand{\PB}{\ensuremath{\mathsf{PB}}}
\newcommand{\Stmt}{\ensuremath{\mathsf{St}}}
\newcommand{\Stmta}{\ensuremath{\mathsf{St1}}}
\newcommand{\scVar}{\ensuremath{v}}
\newcommand{\lpVar}{\ensuremath{\ell}}
\newcommand{\ArVar}{\ensuremath{A}}
\newcommand{\OP}{\ensuremath{\mathsf{op}}}
\newcommand{\BoolE}{\ensuremath{\mathsf{BoolE}}}
\newcommand{\iif}{\ensuremath{\mathbf{if}}}
\newcommand{\eelse}{\ensuremath{\mathbf{else}}}
\newcommand{\tthen}{\ensuremath{\mathbf{then}}}
\newcommand{\ffor}{\ensuremath{\mathbf{for}}}

\newcommand{\cconst}{\ensuremath{\mathsf{c}}}

\newcommand{\AssignStmts}{\ensuremath{\mathsf{AssignSt}}}

\newcommand{\Unlabeled}{\ensuremath{\mathsf{U}}}

\newcommand{\infocomment}[1]{{\scriptsize\ttfamily\textcolor{darkgray}{\newline$\triangleright$ #1}}}
\newcommand{\true}{\ensuremath{\mathsf{True}}}
\newcommand{\false}{\ensuremath{\mathsf{False}}}

\usepackage{xspace}

\makeatletter
\RequirePackage[bookmarks,unicode,colorlinks=true]{hyperref}%
   \def\@citecolor{blue}%
   \def\@urlcolor{blue}%
   \def\@linkcolor{blue}%

\def\orcidID#1{\smash{\href{http://orcid.org/#1}{\protect\raisebox{-1.25pt}{\protect\includegraphics{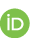}}}}}
\makeatother

\begin{document}

\title{Verifying Array Manipulating Programs with Full-Program Induction}
\author{Supratik Chakraborty\inst{1}\orcidID{0000-0002-7527-7675} \and Ashutosh Gupta\inst{1} \and Divyesh Unadkat\inst{1,2}\orcidID{0000-0001-6106-4719}}

\institute{Indian Institute of Technology Bombay, Mumbai, India\\
  \email{\{supratik,akg\}@cse.iitb.ac.in} \and
  TCS Research, Pune, India\\
  \email{divyesh.unadkat@tcs.com}}

\maketitle
\begin{abstract}
%
%
%
We present a full-program induction technique for proving (a sub-class
of) quantified as well as quantifier-free properties of programs
manipulating arrays of parametric size $N$.
Instead of inducting over individual loops, our technique inducts over
the entire program (possibly containing multiple loops) directly via
the program parameter $N$.
%
Significantly, this does not require generation or use of
loop-specific invariants.
%
%
We have developed a prototype tool {\ourtool}
to assess the efficacy of our technique.
%
We demonstrate the performance of {\ourtool} vis-a-vis several
state-of-the-art tools on a set of array manipulating benchmarks.




\end{abstract}



\section{Introduction}
Programs with loops manipulating arrays are common in a variety of
applications.  Unfortunately, assertion checking in such programs is
undecidable.  Existing tools therefore use a combination of techniques
that work well for certain classes of programs and assertions, and
yield conservative results otherwise.  In this paper, we present a new
technique to add to this arsenal of techniques.  Specifically, we
focus on programs with loops manipulating arrays, where the size of
each array is a symbolic integer parameter $N~(> 0)$.
We allow (a sub-class of) quantified and quantifier-free pre- and
post-conditions that may depend on the symbolic parameter $N$.  Thus,
the problem we wish to solve can be viewed as checking the validity of
a parameterized Hoare triple $\{\varphi(N)\} \;\PP_N\; \{\psi(N)\}$
for all values of $N~(> 0)$, where the program $\PP_N$ computes with
arrays of size $N$, and $N$ is a free variable in $\varphi(\cdot)$ and
$\psi(\cdot)$.  Fig.~\ref{fig:ex}(a) shows an example of one such
Hoare triple, written using {\tt assume} and {\tt assert}.  This
triple effectively verifies that $\sum_{j=0}^{i-1} \left(1 + \sum_{k=0}^{j-1}
6\cdot(k+1)\right) = i^3$ for all $i \in \{0 \ldots N-1\}$, and for
all $N > 0$. 
Although each loop in Fig.~\ref{fig:ex}(a) is simple, their sequential
composition makes it difficult even for state-of-the-art tools like
VIAP~\cite{viap}, VeriAbs~\cite{veriabs}, {\freqhorn}
\cite{madhukar19}, Tiler~\cite{sas17}, Vaphor~\cite{vaphor}, or
Booster~\cite{booster} to prove the post-condition correct. In fact,
none of the above tools succeed in automatically proving the
post-condition in Fig.~\ref{fig:ex}(a).  In contrast, the technique
presented in this paper, called \emph{full-program induction}, proves
the post-condition in Fig.~\ref{fig:ex}(a) correct within a few
seconds.

\begin{figure}[h]
  \begin{tabular}{l|cl}
    \begin{minipage}{0.43\textwidth}
      {\scriptsize
\begin{verbatim}
// assume(true)
1. for (int t1=0; t1<N; t1=t1+1) {
2.   if (t1==0) { A[t1] = 6; }
3.   else { A[t1] = A[t1-1]+6; }
4. }
5. for (int t2=0; t2<N; t2=t2+1) {
6.   if (t2==0) { B[t2] = 1; }
7.   else { B[t2] = B[t2-1]+A[t2-1]; }
8. }
9. for (int t3=0; t3<N; t3=t3+1) {
10.  if (t3==0) { C[t3] = 0; }
11.  else { C[t3] = C[t3-1]+B[t3-1]; }
12.}
// assert(forall i in 0..N-1, C[i]= i^3)
\end{verbatim}
      }

      \begin{center}(a)\end{center}
    \end{minipage}
    &&
    \begin{minipage}{0.6\textwidth}
      {\scriptsize
\begin{verbatim}
// assume(true)
1.  A[0] = 6;
2.  B[0] = 1;
3.  C[0] = 0;
// assert((C[0] = 0^3) and (B[0] = 1^3 - 0^3) and 
//        (A[0] = 2^3 - 2*1^3 + 0^3))
\end{verbatim}
      }
      \vspace{-2em}
      \begin{center}(b)\end{center}
      
      {\scriptsize
\begin{verbatim}
// assume((N > 1) and (C_Nm1[N-2] = (N-2)^3) and
//        (B_Nm1[N-2] = (N-1)^3 - (N-2)^3) and
//        (A_Nm1[N-2] = N^3 - 2*(N-1)^3 + (N-2)^3))
1.  A[N-1] = A_Nm1[N-2] + 6;
2.  B[N-1] = B_Nm1[N-2] + A_Nm1[N-2];
3.  C[N-1] = C_Nm1[N-2] + B_Nm1[N-2];
// assert((C[N-1] = (N-1)^3) and
//        (B[N-1] = N^3 - (N-1)^3) and
//        (A[N-1] = (N+1)^3 - 2*N^3 + (N-1)^3))
\end{verbatim}
      }
      \vspace{-2em}
      \begin{center}(c)\end{center}
      
    \end{minipage}\\
  \end{tabular}
  \vspace{-3ex}
\caption{Original and simplified Hoare triples}
\label{fig:ex}
\end{figure}

Like several earlier approaches~\cite{kinduction}, full-program
induction relies on mathematical induction to reason about programs
with loops.  However, the way in which the inductive claim is
formulated and proved differs significantly. 
Specifically, (i) we \emph{do not require explicit or
  implicit loop-specific invariants} to be provided by the user or
generated by a solver (viz. by constrained Horn clause
solvers~\cite{chc,quic3,madhukar19} or recurrence
solvers~\cite{viap,aligators}), (ii) we \emph{induct on the full
  program} (possibly containing multiple loops) with parameter $N$ and
not on iterations of individual loops in the program, and (iii) we
perform \emph{non-trivial correct-by-construction code
  transformations}, whenever feasible, to simplify the inductive step
of reasoning.  The combination of these factors often reduces
reasoning about a program with multiple loops to reasoning about one
with fewer (sometimes even none) and ``simpler'' loops, thereby
simplifying proof goals.  In this paper, we demonstrate this, focusing
on programs with sequentially composed, but non-nested loops.

As an illustration of simplifications that can result from application
of full-program induction, consider the problem in
Fig.~\ref{fig:ex}(a) again.  Full-program induction reduces checking
the validity of the Hoare triple in Fig.~\ref{fig:ex}(a) to checking
the validity of two ``simpler'' Hoare triples, represented in
Figs.~\ref{fig:ex}(b) and \ref{fig:ex}(c). Note that the programs in
Figs.~\ref{fig:ex}(b) and \ref{fig:ex}(c) are loop-free.  In addition,
their pre- and post-conditions are quantifier-free.  The validity of
these Hoare triples (Figs.~\ref{fig:ex}(b) and \ref{fig:ex}(c)) can
therefore be easily proved, e.g. by bounded model checking~\cite{bmc}
with a back-end SMT solver like {\zthree}~\cite{z3}.  Note that the
value computed in each iteration of each loop in Fig.~\ref{fig:ex}(a)
is data-dependent on previous iterations of the respective loops.
Hence, none of these loops can be trivially translated to a set of
parallel assignments.

Invariant-based techniques, viz.
\cite{Gopan,Halbwachs,Rival,ArrayCousotCL11,Gulwani,Srivastava09,Dirk07,Jhala},
are popularly used to reason about array manipulating programs.  If we
were to prove the assertion in Fig.~\ref{fig:ex}(a) using such
techniques, it would be necessary to use appropriate loop-specific
invariants for each of the three loops in Fig.~\ref{fig:ex}(a). The
weakest loop invariants needed to prove the post-condition in this
example are: $\forall i \in [0...t1-1]\; (A[i] = 6i + 6)$ for the
first loop (lines $1$-$4$), $\forall j \in [0...t2-1]\; (B[j] = 3j^2 +
3j+1) \wedge (A[j] = 6j + 6)$ for the second loop (lines $5$-$8$), and
$\forall k \in [0...t3-1]\; (C[k] = k^3) \wedge (B[k] = 3k^2+3k+1)$
for the third loop (lines $9$-$12$).  Unfortunately, automatically
deriving such quantified non-linear loop invariants is far from
trivial. Template-based invariant generators,
viz.~\cite{houdini,daikon}, are among the best-performers when
generating such complex invariants.  However, their abilities are
fundamentally limited by the set of templates from which they choose.
We therefore choose not to depend on invariants for individual loops
in our work at all.  
Instead of
inducting over the iterations of each individual loop,
we propose to reason about the entire
program (containing one or more loops) directly, while inducting on
the parameter $N$.  
Needless to say, each approach has its own strengths and limitations,
and the right choice always depends on the problem at hand.  Our
experiments show that full-program induction is able to solve several
difficult problem instances with an off-the-shelf SMT solver
({\zthree}) at the back-end, which other techniques either fail to
solve these instances, or rely on sophisticated recurrence solvers.

The primary contributions of our work can be summarized as follows.
\begin{itemize}
\item We introduce the notion of {\em full-program induction} for
  reasoning about assertions in programs with loops manipulating arrays.
\item We present practical algorithms for full-program induction.
\item We describe a prototype tool {\ourtool} that implements the
  algorithms, using an off-the-shelf SMT solver, viz. {\zthree}, at
  the back-end to discharge verification conditions.  {\ourtool}
  outperforms several state-of-the-art tools on a suite of
  array-manipulating benchmark programs.
\end{itemize}



\paragraph{\bfseries Related Work.}

Earlier work on inductive techniques can be broadly categorized into
those that require loop-specific invariants to be provided or
automatically generated, and those that work without them. 
Requiring a ``good'' inductive invariant for every loop
in a program effectively shifts the onus of assertion checking to that
of invariant generation. Among techniques that do not require explicit
inductive invariants or mid-conditions for each loop, there are some
that require loop invariants to be implicitly generated by a
constraint solver.  These include techniques based on constrained Horn
clause solving~\cite{chc,quic3,madhukar19,vaphor}, acceleration and
lazy interpolation for arrays~\cite{booster} and those that use
inductively defined predicates and recurrence
solving~\cite{viap,aligators}, among others.  Thanks to the impressive
capabilities of modern constraint solvers and the effectiveness of
carefully tuned heuristics for stringing together multiple solvers,
this approach has shown a lot of promise in recent years.  However, at
a fundamental level, these formulations rely on solving implicitly
specified loop invariants garbed as constraint solving problems.
There are yet other techniques, such as that in~\cite{lopstr12}, that
truly do not depend on loop invariants being generated.  In fact, the
technique of~\cite{lopstr12} comes closest to our work in principle.
However, \cite{lopstr12} imposes severe restrictions on the input
programs, and the example in Fig.~\ref{fig:ex} does not meet these
restrictions.  Therefore, the technique of~\cite{lopstr12} is
applicable only to a small part of the program-assertion space over
which our technique works.  Techniques such as tiling~\cite{sas17}
reason one loop at a time and apply only when loops have simple data
dependencies across iterations (called \emph{non-interference} of
tiles in~\cite{sas17}).  It effectively uses a slice of the
post-condition of a loop as an inductive invariant, and also requires
strong enough mid-conditions to be generated in the case of
sequentially composed loops.  We circumvent all of these requirements
in the current work.  For some other techniques for analyzing array
manipulating programs, please
see~\cite{ArrayCousotCL11,Jhala,verifast}.





\vspace{-2ex}
\section{Overview of Full-program Induction}
\vspace{-1ex}
\label{sec:overview}
Recall that our objective is to check the validity of the
parameterized Hoare triple $\{\varphi(N)\} \;\PP_{N}\; \{\psi(N)\}$
for all $N > 0$.  At a high level, our approach works like any other
inductive technique.  Thus, we have a base case, where we verify that
the parameterized Hoare triple holds for some small values of $N$, say
$0 < N \le M$.  We then hypothesize that $\{\varphi(N-1)\}
\;\PP_{N-1}\; \{\psi(N-1)\}$ holds for some $N > M$, and try to show
that this implies $\{\varphi(N)\} \;\PP_{N}\; \{\psi(N)\}$.  While
this sounds simple in principle, there are several technical
difficulties en route.  Our contribution lies in overcoming these
difficulties algorithmically for a large class of programs and
assertions, thereby making {\em full-program induction} a viable and
competitive technique for proving properties of array manipulating
programs.

We rely on an important, yet reasonable, assumption that can be stated
as follows: \emph{For every value of $N~(> 0)$, every loop in
  ${\PP_N}$ can be statically unrolled a fixed number (say $f(N)$) of
  times to yield a loop-free program $\widehat{\PP_N}$ that is
  semantically equivalent to ${\PP_N}$.}  Note that this does not
imply that reasoning about loops can be translated into loop-free
reasoning.  In general, $f(N)$ is a non-constant function, and hence,
the number of unrollings of loops in $\PP_N$ may strongly depend on
$N$. In our experience, loops in a vast majority of array manipulating
programs (including Fig.~\ref{fig:ex}(a)) satisfy the above
assumption.  Consequently, the base case of our induction reduces to
checking a Hoare triple for a loop-free program.  Checking such a
Hoare triple is easily achieved by compiling the pre-condition,
program and post-condition into an SMT formula, whose
(un)satisfiability can be checked with an off-the-shelf back-end SMT
solver.

The inductive step is the most complex one, and is the focus of the
rest of the paper.  Recall that the inductive hypothesis asserts
that $\{\varphi(N-1)\}\;{\PP_{N-1}}\;\{\psi(N-1)\}$ is valid.  To make
use of this hypothesis in the inductive step, we must relate the
validity of $\{\varphi(N)\}\;{\PP_N}\;\{\psi(N)\}$ to that of
$\{\varphi(N-1)\}\;{\PP_{N-1}}\;\{\psi(N-1)\}$.  We propose doing
this, whenever possible, via two key notions -- that of ``difference''
program and ``difference'' pre-condition.  Given a parameterized
program ${\PP_N}$, intuitively the ``difference'' program ${\partial
  \PP_N}$ is one such that ${\PP_{N-1}};{\partial \PP_N}$ is
semantically equivalent to ${\PP_N}$, where ``;'' denotes sequential
composition.  It turns out that for our purposes, the semantic
equivalence alluded to above is not really necessary; it suffices to
have ${\partial \PP_N}$ such that
$\{\varphi(N)\}\;{\PP_N}\;\{\psi(N)\}$ is valid iff $\{\varphi(N)\}\;
{\PP_{N-1}};{\partial \PP_N} \;\{\psi(N)\}$ is valid.  We will
henceforth use this interpretation of a ``difference'' program.  The
``difference'' pre-condition ${\partial \varphi(N)}$ is a formula such
that (i) $\varphi(N) \rightarrow (\varphi(N-1) \wedge {\partial
  \varphi(N)})$ and (ii) the execution of $\PP_{N-1}$ doesn't affect
the truth of ${\partial \varphi(N)}$.  Computing ${\partial \PP_N}$
and ${\partial \varphi(N)}$ is not easy in general, and we discuss
this in detail in the rest of the paper.  


Assuming we have ${\partial \PP_N}$ and ${\partial \varphi(N)}$ with
the properties stated above, the proof obligation
$\{\varphi(N)\}\;\PP_N\;\{\psi(N)\}$ can now be reduced to proving
$\{\varphi(N-1)\}\;\PP_{N-1}\;\{\psi(N-1)\}$ and $\{\psi(N-1) \wedge
{\partial \varphi(N)}\} \;{\partial \PP_{N}}\; \{\psi(N)\}$.  The
first triple follows from the inductive hypothesis.  Proving the
second triple may require strengthening the pre-condition, say by a
formula $\ppre(N-1)$, in general.  Recalling that we are in the
inductive step of mathematical induction, we formulate the new proof
sub-goal in such a case as $\{(\psi(N-1) \wedge \ppre(N-1)) \wedge
{\partial \varphi(N)}\} \;{\partial \PP_{N}}\; \{\psi(N) \wedge
\ppre(N)\}$.  While this is somewhat reminiscent of loop invariants,
observe that $\ppre(N)$ is \emph{not} really a loop-specific
invariant. Instead, it is analogous to computing an invariant for the
entire program, possibly containing multiple loops.  Specifically, the
above process strengthens both the pre- and post-condition of
$\{\psi(N-1) \wedge {\partial \varphi(N)}\} \;{\partial \PP_{N}}\;
\{\psi(N)\}$ simultaneously using $\ppre(N-1)$ and $\ppre(N)$,
respectively. The strengthened post-condition of the resulting Hoare
triple may, in turn, require a new pre-condition $\ppre'(N-1)$ to be
satisfied. This process of strengthening the pre- and post-conditions
of the Hoare triple involving $\partial{\PP_N}$ can be iterated until
a fix-point is reached, i.e. no further pre-conditions are needed for
the parameterized Hoare triple to hold.  While the fix-point was
quickly reached for all benchmarks we experimented with, we also discuss
how to handle cases where the above process may not converge easily.  Note
that since we effectively strengthen the pre-condition of the Hoare
triple in the inductive step, for the overall induction to go through,
it is also necessary to check that the strengthened assertions hold at
the end of each base case check.
%
%
The technique described above is called \emph{full-program induction},
and the following theorem guarantees its soundness.
\begin{theorem}
\label{thm:full-prog-ind-sound}
Given $\{\varphi(N)\}\;\PP_N\;\{\psi(N)\}$, suppose the following are
true:
\begin{enumerate}
\item For $N > 1$, $\{\varphi(N)\} \;\PP_{N-1};{\partial \PP_N} \;
  \{\psi(N)\}$ holds iff $\{\varphi(N)\} \;\PP_{N} \; \{\psi(N)\}$
  holds.
\item For $N > 1$, there exists a formula ${\partial \varphi(N)}$ such that
  (a) ${\partial \varphi(N)}$ doesn't refer to any program variable or
array element modified in $\PP_{N-1}$, and
(b) $\varphi(N) \rightarrow \varphi(N-1) \wedge {\partial \varphi(N)}$.
\item There exists an integer $M \ge 1$ and a parameterized formula $\ppre(M)$
  such that
  (a) $\{\varphi(N)\}\;\PP_N\;\{\psi(N)\}$ holds for $0 < N \le M$,
  (b) $\{\varphi(M)\}\;\PP_M\;\{\psi(M)\wedge\ppre(M)\}$ holds, and
  (c) $\{\psi(N-1)\wedge \ppre(N-1) \wedge
  {\partial \varphi(N)}\}\;{\partial \PP_N}\; \{\psi(N) \wedge \ppre(N)\}$ holds
  for $N > M$.
\end{enumerate}
Then $\{\varphi_N\}\;\PP_N\;\{\psi_N\}$ holds for all $N \ge 1$.
\end{theorem}
\begin{proof}
For $0 < N \le M$, condition 3(a) ensures that
$\{\varphi(N)\}\;\PP_N\;\{\psi(N)\}$ holds.  For $N > M$, note that by
virtue of condition 1 and 2(b), $\{\varphi(N)\}\;\PP_N\;\{\psi(N)\}$
holds if $\{\varphi(N-1) \wedge {\partial
  \varphi(N)}\}\;\PP_{N-1};{\partial \PP_N}\;\{\psi(N) \wedge
\ppre(N)\}$ holds. With $\psi(N-1) \wedge \ppre(N-1)$ as a
mid-condition, and by virtue of condition 2(a), the latter Hoare
triple holds for $N > M$ if $\{\varphi(M)\}\;\PP_{M}\;\{\psi(M) \wedge
\ppre(M)\}$ holds and $\{\psi(N-1) \wedge \ppre(N-1) \wedge {\partial
  \varphi(N)}\}\;{\partial \PP_N}\;\{\psi(N) \wedge \ppre(N)\}$ holds
for all $N > M$. Both these triples are seen to hold by virtue of
conditions 3(b) and (c). \qed
\end{proof}

\vspace{-4ex}
\section{Algorithms to perform Full-program Induction}
\vspace{-1ex}
\label{sec:details}
We now discuss the {\em full-program induction}
algorithm, focusing on generation of three crucial components:
\emph{difference program} ${\partial \PP_N}$, \emph{difference
  pre-condition} ${\partial \varphi(N)}$, and the formula $\ppre(N)$
for strengthening pre- and post-conditions.

\vspace{-.2in}
\subsection{Preliminaries}
\label{sec:prelim}
We consider array manipulating programs generated by the grammar shown
below (adapted from~\cite{sas17}).  
\begin{center}
\begin{tabular}{rcl}
  {\PB} & ::= & \Stmt\\
  \Stmt & ::= & {\scVar} := \EE ~$\mid$~ {\ArVar}[\EE] := \EE ~$\mid$~
                {\iif}(\BoolE) {\tthen} {\Stmt} {\eelse} \Stmt ~$\mid$~
                \Stmt~;~\Stmt ~$\mid$~\\
        &     & {\ffor} ({\lpVar} := 0; {\lpVar} $<$ \EE; {\lpVar} := {\lpVar}+1) ~\{{\Stmta}\}\\
                \Stmta & ::= & {\scVar} := \EE ~$\mid$~ {\ArVar}[\EE] := \EE  ~$\mid$~
                       {\iif}(\BoolE) {\tthen} {\Stmta} {\eelse} \Stmta
                       ~$\mid$~ \Stmta~;~\Stmta \\
 \EE & ::=  &\EE ~\OP~ \EE ~$\mid$~ {\ArVar}[\EE] ~$\mid$~ {\scVar} ~$\mid$~ {\lpVar} ~$\mid$~ {\cconst} ~$\mid$~ $N$ \\
 \OP & ::= & + ~$\mid$~ - ~$\mid$~ * ~$\mid$~ / \\
 \BoolE & ::= & \EE ~$\mathsf{relop}$ \EE ~$\mid$~ {\BoolE} $\mathsf{AND}$ {\BoolE} ~$\mid$~ $\mathsf{NOT}$ {\BoolE} ~$\mid$~ {\BoolE} $\mathsf{OR}$ {\BoolE}
\end{tabular}
\end{center}

\noindent
This grammar restricts programs to have non-nested loops.  While this
limits the set of programs to which our technique currently applies,
there is a large class of useful programs, with possibly long
sequences of loops, that are included in the scope of our work.  In
reality, our technique also applies to a subclass of programs with
nested loops.  However, characterizing this class of programs through
a grammar is a bit unwieldy, and we avoid doing so for reasons of
clarity. A program $\PP_N$ is a tuple $(\mathcal{V}, \mathcal{L},
\mathcal{A}, {\PB}, N)$, where $\mathcal{V}$ is a set of scalar
variables, $\mathcal{L} \subseteq \mathcal{V}$ is a set of scalar loop
counter variables, $\mathcal{A}$ is a set of array variables, ${\PB}$
is the program body, and $N$ is a special symbol denoting a positive
integer parameter. In the grammar shown above, we assume ${\ArVar} \in
\mathcal{A}$, ${\scVar} \in \mathcal{V} \setminus \mathcal{L}$,
${\lpVar} \in \mathcal{L}$ and ${\cconst}\in \integers$.  Furthermore,
``$\mathsf{relop}$'' is assumed to be one of the relational operators
and ``$\mathsf{op}$''is an arithmetic operator from the set \{+, -, *,
/\}.  We also assume that each loop $L$ has a unique loop counter
variable $\ell$ which is initialized at the beginning of $L$ and is
incremented by $1$ at the end of each iteration. Assignments in the
body of $L$ are assumed not to update $\ell$.
Finally, for each loop with termination condition $\ell < \EE$, we
assume that $\EE$ is an expression in terms of $N$. We denote by
$k_L(N)$ the number of times loop $L$ iterates in the program with
parameter $N$. We verify Hoare triples of the form $\{\varphi(N)\}
\;\PP_N\; \{\psi(N)\}$, where $\varphi(N)$ and $\psi(N)$ are either
universally quantified formulas of the form $\forall I\, \left(\Phi(I,
N) \implies \Psi(\mathcal{A}, \mathcal{V}, I, N)\right)$ or
quantifier-free formulas of the form $\Xi(\mathcal{A}, \mathcal{V},
N)$. In the above, $I$ is a sequence of array index variables, $\Phi$
is a quantifier-free formula in the theory of arithmetic over
integers, and $\Psi$ and $\Xi$ are quantifier-free formulas in the
combined theory of arrays and arithmetic over integers. 

Static single assignment (SSA)~\cite{ssa} is a well-known technique
for renaming scalar variables such that a variable is written at most
once in a program.  For our purposes, we also wish to rename arrays so
that each loop updates its own version of an array and multiple writes
to an array element within the same loop happen on different versions
of the array.  Array SSA~\cite{arrayssa} renaming has been studied
earlier in the context of compilers to achieve this goal.  We propose
using SSA renaming for both scalars and arrays as a pre-processing
step of our analysis.  Therefore, we assume henceforth that the input
program is SSA renamed (for both scalars and arrays).  We also assume
that the post-condition is expressed in terms of these SSA renamed
scalar and array variables.

We represent a program 
using a {\em control flow graph} $G = (Locs, Edges, \mu)$, where
$Locs$ denotes the set of control locations (nodes) of the program,
$Edges \subseteq Locs \times Locs \times \{\ltrue, \lfalse,
\Unlabeled\}$ represents the flow of control and $\mu: Locs
\rightarrow {\AssignStmts}$ $\union$ ${\BoolE}$ annotates every node
in $Locs$ with either an assignment statement (of the form ${\scVar}
:= \EE$ or ${\ArVar}[\EE] := \EE$) from the set of assignment
statements {\AssignStmts}, or a Boolean condition.  Two distinguished
control locations, called $\mathsf{Start}$ and $\mathsf{End}$ in
$Locs$, represent the entry and exit points of the program.  An edge
$(n_1, n_2, label)$ represents flow of control from $n_1$ to $n_2$
without any other intervening node. It is labeled $\ltrue$ or
$\lfalse$ if $\mu(n_1)$ is a Boolean condition, and is labeled
$\Unlabeled$ otherwise.  If $\mu(n_1)$ is a Boolean condition, there
are two outgoing edges from $n_1$, labeled $\ltrue$ and $\lfalse$
respectively, and control flows from $n_1$ to $n_2$ along $(n_1, n_2,
label)$ only if $\mu(n_1)$ evaluates to $label$. If $\mu(n_1)$ is an
assignment statement, there is a single outgoing edge from $n_1$, and
it is labeled $\Unlabeled$.  Henceforth, we use CFG to refer to the
control flow graph.

A CFG may have cycles due to the presence of loops in the program.
A \emph{back-edge} of a loop is an edge from the node corresponding to
the last statement in the loop body to the node representing the loop head.
An \emph{exit-edge} is an edge from the loop head to a node outside the
loop body. An \emph{incoming-edge} is an edge to the loop head from a
node outside the loop body. We assume that every loop has exactly one
\emph{back-edge}, one \emph{incoming-edge} and one \emph{exit-edge}.
For technical reasons, and without loss of generality, we also assume that
the \emph{exit-edge} of a loop always goes to a ``nop'' node (say, having
a statement {\tt x = x;}).

Given a program, the program dependence graph (or PDG) $G = (V, DE,
CE)$ represents data and control dependencies among program
statements. Here, $V$ denotes vertices representing assignment
statements and boolean expressions, $DE \subseteq V \times V$ denotes
data dependence edges and $CE \subseteq V \times V$ denotes control
dependence edges. Standard dataflow analysis identifies dependencies
between program variables and thereby among statements. Dependence
between statements updating array elements requires a more careful
analysis.  Let $S_1$ and $S_2$ be two statements in loops $L_1$ and
$L_2$ where there is a control-flow path from $S_1$ to $S_2$ in the
CFG. Suppose $S_1$ is of the form $A[f(i_1, N)] = F(\ldots);$ where
$f$ is an array index expression, $i_1$ is the loop counter variable
of $L_1$, and $F$ is an arbitrary expression. Suppose $S_2$ is of the
form $X = G(A[g(i_2, N)]);$, where $X$ is a variable or array element,
$G$ is an arbitrary expression, and $g$ is an array index expression.
\begin{definition}
\label{def:loopdepends}
We say that $S_2$ in $L_2$ \textbf{depends} on $S_1$ in $L_1$ if there
exists $i_1, i_2$ such that $0 \leq i_1 < k_{L_1}(N)$ and $0 \leq i_2
< k_{L_2}(N)$ and $f(i_1, N) = g(i_2, N)$.
\end{definition}



\begin{algorithm}[!t]
  \caption{\footnotesize \textsc{ComputeRefinedPDG}($\PP_N$ : Program)}
  \label{alg:compute-pdg}
  \scriptsize
  \begin{algorithmic}[1]
    \State $G(V, DE, CE)$ := \textsc{ConstructPDG}($\PP_N$);
    \If{$\exists v, n, n'.\; (n,n') \in DE \wedge is\text{-}array(v) \wedge v \in \mathit{def}(n) \wedge v \in \mathit{uses}(n')$}
      \If{$n$ is part of a loop $L$}
        \State $\ell$ := loop counter of $L$;
        \State Let $\phi(n)$ be the constraint $(0 \leq \ell < k_{L})$;
      \Else
        \State Let $\phi(n)$ be $true$;
      \EndIf
      \If{$n'$ is part of a loop $L'$}
        \State $\ell'$ := loop counter of $L'$;
        \State Let $\phi'(n')$ be the constraint $(0 \leq \ell' < k_{L'})$;
      \Else
        \State Let $\phi'(n')$ be $true$;
      \EndIf      
      \If{$\phi(n) \wedge \phi(n') \wedge \left(\mathit{subscript}(v, n) = \mathit{subscript}(v,n')\right)$ is unsatisfiable}
        \State $DE$ = $DE \setminus \{(n, n')\}$;  \Comment{Remove dependence edges with non-overlapping subscripts}
      \EndIf
    \EndIf
    \State \Return $G(V, DE, CE)$;
  \end{algorithmic}
\end{algorithm}



The routine \textsc{ComputeRefinedPDG} shown in Algorithm
\ref{alg:compute-pdg} constructs and refines the program dependence
graph $G = (V, DE, CE)$ for the input program $\PP_N$. It uses the
function \textsc{ConstructPDG} (line 1) based on the technique of
\cite{pdg} to create an initial graph. For a node $n$ in $G$, let
$\mathit{def}(n)$ and $\mathit{uses}(n)$ refer to the set of
variables/array elements defined and used, respectively, in the
statement/boolean expression corresponding to $n$. Similarly, let
$\mathit{subscript}(v,n)$ refer to the index expression of the array
element $v$ referred to at node $n$. Predicate $is\text{-}array(v)$
evaluates to true if the $v$ is an array element and false if $v$ is a
scalar variable. Note that lines 2-14 of \textsc{ComputeRefinedPDG}
removes data dependence edges between nodes of $G$ that do not satisfy
Definition~\ref{def:loopdepends}. 




\vspace{-.15in}
\subsection{Core Modules in the Technique}

\paragraph{\bfseries Peeling the Loops.}
To relate $\PP_N$ to $\PP_{N-1}$, we first ensure that the
corresponding loops in both programs iterate the same number of times
by \emph{peeling} extra iterations from the loops in $\PP_N$. This is
done by routine \textsc{PeelAllLoops} shown in Algorithm
\ref{alg:peelloops}. The algorithm first makes a copy, viz. $\PP^p_N$,
of the input CFG $\PP_N$.  Let $\textsc{Loops}(\PP^p_N)$ denote the
set of loops of $\PP^p_N$, and let $k_L(N)$ and $k_L(N-1)$ denote the
number of times loop $L$ iterates in $\PP^p_N$ and $\PP^p_{N-1}$
respectively. The difference $k_L(N) - k_L(N-1)$, computed in line 5,
gives the extra iterations of loop $L$ in $\PP^p_N$. If this difference
is not a constant, we currently report a failure of our technique
(line 6). Otherwise, routine \textsc{PeelSingleLoop} transforms loop
$L$ of $\PP^p_N$ as follows: it replaces the termination condition
$(\ell < k_L(N))$ of $L$ by $(\ell < k_L(N-1))$.  It also peels (or
unrolls) the last $(k_L(N) - k_L(N-1))$ iterations of $L$ and adds
control flow edges such that the the peeled iterations are executed
immediately after the loop body is iterated $k_L(N-1)$ times.
Effectively, \textsc{PeelSingleLoop} unrolls/peels the last
$(k_L(N)-k_L(N-1))$ iterations of loop $L$ in $\PP^p_N$. The
transformed CFG is returned as the updated $\PP^p_N$ in line 7.
In addition, \textsc{PeelSingleLoop} also returns the set $Locs'$
of all CFG nodes newly added while peeling the loop $L$.  
The overall updated CFG and the set of all peeled nodes obtained
after peeling all loops in $\PP^p_N$ is returned in line 9.

\begin{algorithm}[!t]
  \caption{\footnotesize \textsc{PeelAllLoops}$\left((Locs, Edges, \mu): \mbox{ CFG of } \PP_N \right)$}
  \label{alg:peelloops}
  \scriptsize
  \begin{algorithmic}[1]
  \State $\PP^p_N := (Locs^p, Edges^p, \mu^p)$, where $L^p = Locs$, $Edges^p = Edges$, $\mu^p = \mu$; \Comment{Copy of $\PP_N$}
  \State $peelNodes$ := $\varnothing$; 
  \For{each loop $L \in \textsc{Loops}( \PP^p_N )$}
  \label{alg:line:canon-peel-loop}
    \State Let $k_L({N})$ be the expression for iteration count of $L$ in $\PP^p_N$;
    \State $peelCount := \textsc{Simplify}(k_L({N})-k_L({N-1}))$;
    \If{$peelCount$ is non-constant}
     {\bf throw} ``Failed to peel non-constant number of iterations'';
     \EndIf
     \State $\langle\PP^p_N, Locs'\rangle := \textsc{PeelSingleLoop}(\PP^p_N, L, k_L({N-1}), peelCount)$;
     \infocomment{Transforms loop $L$ so that last $peelCount$ iterations of $L$ are peeled/unrolled. Updated CFG and newly created CFG nodes for the peeled iterations are returned by \textsc{PeelSingleLoop}.}
     \State $peelNodes$ := $peelNodes$ $\union$ $Locs'$;
     \EndFor
     \State \Return $\langle\PP^p_N, peelNodes\rangle$;
     \vspace{-1ex}
  \end{algorithmic}
\end{algorithm}





\vspace{-1ex}
\begin{lemma}
  \label{lemma:alg-peelloops}
  $\{\varphi_N\}\;\PP_N\;\{\psi_N\}$ holds iff $\{\varphi_N\}\;\PP^p_N\;\{\psi_N\}$ holds.
\end{lemma}
\vspace{-1ex}

\vspace{-.1in}

\paragraph{\bfseries Affected Variable Analysis.}
Before we discuss the generation of ${\partial \PP_N}$, we present an
analysis that identifies variables/array elements that may take
different values in $\PP_N$ and $\PP_{N-1}$.  For example, the first
$k_L(N-1)$ iterations of $L$ in $\PP_N$ may not be semantically
equivalent to the (entire) $k_L(N-1)$ iterations of $L$ in
$\PP_{N-1}$. This is because the semantics of statements in $L$ may
depend on the value of $N$ either directly or indirectly.  We call
variables/array elements updated in such statements as \emph{affected}
variables.
For every loop with statements having potentially different semantics
in $\PP_N$ and $\PP_{N-1}$, the difference program ${\partial \PP_N}$
must have a version of the loop with statements that restore the
effect of the first $k_L(N-1)$ iterations of $L$ in $\PP_N$ after the
(entire) $k_L(N-1)$ iterations of $L$ in $\PP_{N-1}$ have been
executed. Furthermore, for statements in $\PP_N$ that are not enclosed
within loops but have potentially different semantics from the
corresponding statements in $\PP_{N-1}$, ${\partial \PP_N}$ must also
rectify the values of variables/array elements
updated in such statements.




\begin{algorithm}[!t]
  \caption{\footnotesize \textsc{ComputeAffected}($\PP_N$ : Program, $peelNodes$ : Peeled Statements)}
  \label{alg:pdg-affected}
  \scriptsize
  \begin{algorithmic}[1]
    \State $G(V, DE, CE)$ := \textsc{ComputeRefinedPDG}($\PP_N$);
    \State $\mathsf{AffectedVars} := \{N\}$;  \Comment{$N$ is in the affected set}
    \Repeat
      \State $\mathsf{WorkList}$ := $V \setminus peelNodes$; \Comment{all non-peeled nodes in $G$}
      \While{$\mathsf{WorkList} \not = \{\}$}
        \State Remove a node $n$ from $\mathsf{WorkList}$;
        \If{$\exists v.\; is\text{-}array(v) \wedge (\exists u. \; u \in subscript(v,n) \wedge u \in \mathsf{AffectedVars})$} 
          \State $\mathsf{AffectedVars}$ := $\mathsf{AffectedVars}$ $\cup$ $v$;
        \EndIf
         \If{$\exists v.\; v \in uses(n)$} 
           \If{$\exists m.\; m \in \mathit{reaching}\text{-}\mathit{def}(v, n) \wedge m \in peelNodes$} 
             \State $\mathsf{AffectedVars}$ := $\mathsf{AffectedVars}$ $\cup$ $def(n)$;
           \EndIf
           \If{$\exists m.\; m \in \mathit{reaching}\text{-}\mathit{def}(v, n) \wedge def(m) \in \mathsf{AffectedVars}$} 
             \State $\mathsf{AffectedVars}$ := $\mathsf{AffectedVars}$ $\cup$ $def(n)$;
           \EndIf
           \If{$v \in \mathsf{AffectedVars} \wedge n$ is a assignment node}
             \State $\mathsf{AffectedVars}$ := $\mathsf{AffectedVars}$ $\cup$ $def(n)$;
           \EndIf
           \If{$v \in \mathsf{AffectedVars} \wedge n$ is a predicate node}
             \For{each edge $(n,n') \in CE$}
               \State $\mathsf{AffectedVars}$ := $\mathsf{AffectedVars}$ $\cup$ $def(n')$;
             \EndFor
           \EndIf
         \EndIf
       \EndWhile
    \Until{$\mathsf{AffectedVars}$ ~does ~not ~change}
    \State \Return $\mathsf{AffectedVars}$;
  \end{algorithmic}
\end{algorithm}






Subroutine \textsc{ComputeAffected}, shown in Algorithm
\ref{alg:pdg-affected}, computes the set of \emph{affected} variables
$\PP_N$.  We first construct the program dependence graph by calling
the function \textsc{ComputeRefinedPDG} (line 1) defined in Algorithm
\ref{alg:compute-pdg}.  Let $\mathsf{AffectedVars}$ represent the set
of \emph{affected} variables/array elements.  We initialize it (line
2) with variable $N$ since its value is different in $\PP_N$ and
$\PP_{N-1}$.
For a node $n$ in the PDG $G$, we use
$\mathit{reaching}\text{-}\mathit{def}(v,n)$ to refer to the set of
nodes where the variable/array element $v$ is defined and the
definition reaches its use at node $n$. In line 4, we collect nodes in
the graph that are not the ones peeled from loops in $\PP_N$. The loop
in lines 5-18 iterates over the collected nodes to identify affected
variables. If a variable in the index expression of an array access is
affected then that array element is considered affected (lines 7-8). A
definition at a node $n$ is affected (marked in line 11) if any
variable $v$ used in the statement (checked in line 9) is defined in a
\emph{peeled} node (line 10). Similarly if the reaching definition of
$v$ is affected (line 12) the definition at $n$ is affected (line
13). A variable defined in terms of an affected variable is also
deemed to be affected (lines 14-15). Finally, a variable definition
that is control dependent on an affected variable is also considered
affected (lines 16-18).  The computation of affected variables is
iterated until the set $\mathsf{AffectedVars}$ saturates.

\vspace{-1ex}
\begin{lemma}
\label{lemma:dfa-sound}
Variables/Array elements not present in $\mathsf{AffectedVars}$ have
the same value after $k_L(N-1)$ iterations of its enclosing
loop (if any) in $P_{N-1}$ as in $P_N$.
\end{lemma}




\vspace{-.1in}

\paragraph{\bfseries Generating the Difference Program $\mathbf{\partial \PP_N}$.}

The routine \textsc{ProgramDiff} in Algorithm~\ref{alg:diff-generation}
shows how the difference program is computed.
We peel each loop in the program and collect the list of peeled nodes (line 1)
using Algorithm~\ref{alg:peelloops}.
We then compute the set of \emph{affected} variables (line 2) using Algorithm
\ref{alg:pdg-affected}.
The difference program $\partial \PP_N$ inherits the skeletal structure of
the program $\PP_N$ after peeling each loop (line 4).
The algorithm then traverses the CFG of each loop in $\PP_N$ and
removes the loops (lines 16-17) that do not update any \emph{affected
}variables from ${\partial \PP_N}$.  For every CFG node in other
loops, it determines the corresponding node type (assignment or
branch) and acts accordingly (lines 7-14).  To explain the intuition
behind the steps of this algorithm, we use the convention that all
variables and arrays of $\PP_{N-1}$ have the suffix {\tt \_Nm1} (for
N-minus-1), while those of $\PP_N$ have the suffix {\tt \_N}.  This
allows us to express variables/array elements of $\PP_N$ in terms of
the corresponding variables/array elements of $\PP_{N-1}$ in a
systematic way in ${\partial \PP_N}$, given that the intended
composition is $\PP_{N-1};{\partial \PP_N}$.

For assignment statements using simple arithmetic operators ({\tt
  +,-,*,/}), the sub-routine \textsc{AssignmentDiff} in
Algorithm~\ref{alg:diff-generation} computes a ``difference''
statement as follows.  We assume that \textsc{Nodes}($L$) returns the
set of CFG nodes in loop $L$.  For every assignment statement of the
form {\tt v = E;} in $L$, a corresponding statement is generated in
${\partial \PP_N}$ that expresses {\tt v\_N} in terms of {\tt v\_Nm1}
and the difference (or ratio) between versions of variables/arrays
that appear as sub-expressions in {\tt E} in $\PP_{N-1}$ and $\PP_N$.
For example, the statement {\tt A\_N[i] = B\_N[i] + v\_N;} in
$\PP_{N}$ gives rise to the ``difference'' statement {\tt A\_N[i] =
  A\_Nm1[i] + (B\_N[i] - B\_Nm1[i]) + (v\_N - v\_Nm1);} in ${\partial
  \PP_N}$.  Similarly, the statement {\tt A\_N[i] = B\_N[i] * v\_N;}
in $\PP_{N}$ gives rise to the ``difference'' statement {\tt A\_N[i] =
  A\_Nm1[i] * (B\_N[i] / B\_Nm1[i]) * (v\_N / v\_Nm1);} under the
assumption {\tt B\_Nm1[i] * v\_Nm1} $\neq 0$.

\begin{algorithm}[!t]
  \caption{\footnotesize \textsc{ProgramDiff}($\PP_N$: program)}
  \label{alg:diff-generation}
  \scriptsize
  \begin{algorithmic}[1]
    \State $\langle\PP_N, peelNodes\rangle$ := \textsc{PeelAllLoops}($\PP_N$);
    \State $\mathsf{AffectedVars}$ := \textsc{ComputeAffected}$(\PP_N, peelNodes)$; 
    \State Let the CFG of $\PP_{N}$ be $(Locs, E, \mu)$; 
    \State $ \partial \PP_N := (Locs', E', \mu') $,
    where $Locs' := Locs$, $E' := E$, and $\mu' := \emptyset$;
    \For{each loop $L \in \textsc{Loops}( \PP_N )$}

   \If{ $\exists v$ such that $v$ is updated in $L$ and $v \in \mathsf{AffectedVars}$}
   \For{each node $n \in \textsc{Nodes}(L)$}
   \State $st_{N} := \mu(n)$;
   \If{$st_{N}$ is of the form $w_{N} := r^1_{N}$ {\tt op} $r^2_{N}$}
     \State $\mu'(n) := $ \textsc{AssignmentDiff}( $w_{N} := r^1_{N}$ {\tt op} $r^2_{N}$ );
   \ElsIf{$st_{N}$ is of the form $w_{N} := w_{N}$ {\tt op} $r^1_{N}$ wherein $w_{N}$ is a scalar}
     \State $\mu'(n) := $ \textsc{AggregateAssignmentDiff}( $L$, $w_{N} := w_{N}$ {\tt op} $r^1_{N}$ );
   \Else \Comment{$st_{N}$ is a conditional statement}
      \State $\mu'(n) := $ \textsc{BranchDiff}( $st_{N}$, $\mathsf{AffectedVars}$ );
    \EndIf
    \EndFor
    \Else \Comment{Remove loop $L$ from CFG of ${\partial \PP_N}$}
    \State $(n_1,n,\Unlabeled) := \textsc{IncomingEdge}(L)$;  $(n,n_2,\lfalse):= \textsc{ExitEdge}(L)$;
    \State $E' := E' \setminus \{(n_1,n, \Unlabeled), (n, n_2, \lfalse)\}$ $\union$ $\{(n_1, n_2,\Unlabeled)\}$;  $Locs' := Locs' \setminus \mathsf{Nodes}(L)$;
    \EndIf
    \EndFor

    \State \Return ${\partial \PP_N}$;
  \end{algorithmic}

  \vspace{2ex}
  \textsc{AssignmentDiff}( $w_{N} := r^1_{N}$ {\tt op} $r^2_{N}$ )
  \begin{algorithmic}[1]
    \State Let {\tt invop} be the arithmetic inverse operator of {\tt op};
    \infocomment{$+$ and $-$ are inverse operators of each other, and so are $\times$ and $\div$}
    \If{{\tt op} $\in \{+, \times\}$}
    \State \Return $w_N$ := $w_{Nm1}$ {\tt op} ($\textsc{Simplify}$($r^1_N$ {\tt invop} $r^1_{Nm1}$) {\tt op} $\textsc{Simplify}$($r^2_N$ {\tt invop} $r^2_{Nm1}$));
    \ElsIf{{\tt op} $\in \{-, \div \}$}
    \State \Return $w_N$ := $w_{Nm1}$ {\tt invop} ($\textsc{Simplify}$($r^1_N$ {\tt op} $r^1_{Nm1}$) {\tt op} $\textsc{Simplify}$($r^2_N$ {\tt op} $r^2_{Nm1}$));
    \Else
    \State {\bf throw} ``Specified operator not handled'';
    \EndIf    
  \end{algorithmic}

  \vspace{1ex}
  \textsc{AggregateAssignmentDiff}( $L$: loop, $w_{N} := w_{N}$ {\tt op} $r^1_{N}$ )
  \begin{algorithmic}[1]
      \State $n_{fresh}  := \textsc{FreshNode}()$;  $\mu'(n_{fresh})$ := ($w_N := w_{Nm1}$);  $Locs' := Locs'$ $\union$ $\{n_{fresh}\}$;
      \State $ (n',n'',\Unlabeled):= \textsc{IncomingEdge}(L)$;
      \State $E' := E' \setminus \{(n',n'',\Unlabeled)\}$ $\union$ $\{(n',n_{fresh}, \Unlabeled),(n_{fresh},n'',\Unlabeled)\}$;
      \If{{\tt op} $\in \{+, *\}$}
	  \State \Return $w_N$ := $w_N$ {\tt op} $\textsc{Simplify}$($r^1_N$ {\tt invop} $r^1_{Nm1}$);
      \ElsIf{{\tt op} $\in \{-, \div \}$}
           \State \Return $w_N$ := $w_N$ {\tt op} $\textsc{Simplify}$($r^1_N$ {\tt op} $r^1_{Nm1}$);
      \Else
           \State {\bf throw} ``Specified operator not handled'';
      \EndIf    
  \end{algorithmic}

  \vspace{1ex}
  \textsc{BranchDiff}( $st_N$: branch condition, $\mathsf{AffectedVars}$ : set of affected variables )
  \begin{algorithmic}[1]
    \State Let $n$ be CFG node corresponding to $st_N$;
    \If{($\exists v$ such that v is read in $st_N$ and $v \in \mathsf{AffectedVars}$) $\vee$ ($st_N \not= st_{N-1}$ is satisfiable)}
      \State {\bf throw} ``Branch conditions in $\PP_N$ and $\PP_{N-1}$ may not evaluate to same value'';
    \Else
      \State \Return $st_{N-1}$;
    \EndIf    
  \end{algorithmic}

\end{algorithm}



There are additional kinds of statements that need special processing
when generating ${\partial \PP_N}$.  These relate to accumulation of
differences (or ratios).  For example, if $\PP_{N}$ has a loop {\tt for(i = 0;
  i < N; i++)  sum\_N = sum\_N + A\_N[i]; } then the
difference {\tt A\_N[i] - A\_Nm1[i]} is aggregated over all indices
from $0$ through $N-2$.  In this case, the corresponding
``difference'' loop in ${\partial \PP_N}$ has the following form: {\tt
  sum\_N = sum\_Nm1; for (i = 0; i < N-1; i++) { sum\_N = sum\_N +
    (A\_N[i] - A\_Nm1[i]);}}. A similar aggregation for multiplicative
ratios can also be defined.  Sub-routine
\textsc{AggregateAssignmentDiff} in Algorithm~\ref{alg:diff-generation} 
generates these ``difference'' statements.

Note that expressions like {\tt (B\_N[i] - B\_Nm1[i])} or {\tt
  (v\_N/v\_Nm1)} can often be simplified from the already generated
part of ${\partial \PP_N}$.  For example, if the already generated
part has a statement of the form {\tt B\_N[i] = B\_Nm1[i] + expr1;} or
{\tt v\_N = expr2*v\_Nm1;}, and if {\tt expr1} and {\tt expr2} are
constants or functions of $N$ and loop counters, then we can use {\tt
  expr1} for {\tt B\_N[i] - B\_Nm1[i]} and {\tt expr2} for {\tt
  v\_N/v\_Nm1} respectively.  We use these optimizations aggressively
in the function $\textsc{Simplify}$ used in \textsc{AssignmentDiff}
and \textsc{AggregateAssignmentDiff}.
  
\begin{algorithm}[!t]
  \caption{\footnotesize \textsc{SimplifyDiff}($\partial \PP_N$: difference program)}
  \label{alg:simpdiff}
  \scriptsize
  \begin{algorithmic}[1]
    \State $ \partial \PP_N := (Locs, E, \mu)$
    \State $ \partial \PP'_N := (Locs', E', \mu')$,
    where $Locs' := Locs$, $E' := E$, and $\mu' := \mu$;
    \For{each loop $L \in \textsc{Loops}(\partial \PP_N)$}
          \State $(n_1,n,\Unlabeled) := \textsc{IncomingEdge}(L)$;  $(n,n_2,\lfalse):=\textsc{ExitEdge}(L)$;
        \If{Loop body of $L$ is of the form $w_{N} := w_{N}$ {\tt op} $expr$,
        wherein $w_{N}$ is a scalar variable}
          \State $n_{acc}$ = $\textsc{FreshNode}()$;
          \If{{\tt op} $\in \{ +, -\}$}
            \State $\mu'(n_{acc})$ := ($w_{N} := w_{N}$ {\tt op} $\textsc{Simplify}(k_L(N-1) * expr)$);
          \ElsIf{{\tt op} $\in \{*, \div \}$}
            \State $\mu'(n_{acc})$ := ($w_{N} := w_{N}$ {\tt op} $\textsc{Simplify}(expr^{k_L(N-1)})$);
          \Else {} {\bf throw} ``Specified operator not handled'';
          \EndIf
          \State $E' := E'$ - $\{(n_1,n, \Unlabeled), (n, n_2, \lfalse)\}$ $\union$ $\{(n_1, n_{acc},\Unlabeled), (n_{acc}, n_2, \Unlabeled)\}$;
          \State $Locs'$ := $Locs' - \textsc{Nodes}(L)$ $\union$ $\{ n_{acc} \}$ ;
        \EndIf
        \If{Loop body of $L$ is of the form $w_{N} := w_{Nm1}$ or $w_{N} := w_{N}$}
          \State $E' := E' - \{(n_1,n, \Unlabeled), (n, n_2, \lfalse)\}$ $\union$ $\{(n_1, n_2,\Unlabeled)\}$;  $Locs' := Locs' - \textsc{Nodes}(L)$;
        \EndIf
    \EndFor
    \State \Return $\partial \PP'_N$
  \end{algorithmic}
\end{algorithm}

For every CFG node representing a conditional branch in $\PP_{N}$,
Algorithm \textsc{BranchDiff} is used to determine if the result of
the condition check can differ in $\PP_N$ and $\PP_{N-1}$.  If not,
the conditional statement can be retained as such in the
``difference'' program.  Otherwise, our current technique cannot
compute ${\partial \PP_N}$ and we report a failure of our technique
(see body of \textsc{BranchDiff}). For example, the conditional
statement {\tt if (t3 == 0) } in line 10 of Fig.~\ref{fig:ex}(a)
behaves identically in $\PP_{N-1}$ and $\PP_N$, and therefore can be
used as is in the loop in the difference program.


\vspace{-.7ex}
\begin{lemma}
\label{lemma:diff-gen-sound}
$\partial \PP_N$ generated by \textsc{ProgramDiff}
is such that, for all $N > 1$,
$\{\varphi(N)\} \;\PP_{N-1};{\partial \PP_N} \; \{\psi(N)\}$ holds iff
$\{\varphi(N)\} \;\PP_{N} \; \{\psi(N)\}$ holds.
\end{lemma}

\paragraph{\bfseries Simplifying the Difference Program.}
\label{sec:simp-diff}

While we have described a simple strategy to generate ${\partial
  \PP_N}$ above, this may lead to redundant statements in the naively
generated ``difference'' code.  For example, we may have a loop like
{\tt for (i=0; i < N-1; i++) { A\_N[i] = A\_Nm1[i];}}.  Our
implementation aggressively optimizes and removes such redundant code,
renaming variables/arrays as needed (see routine \textsc{SimplifyDiff}
in Algorithm~\ref{alg:simpdiff}).  The program ${\partial \PP_N}$ may
also contain loops that compute values of variables that can be
accelerated. For example, we may have a loop {\tt for (i=0; i < N-1;
  i++) { sum = sum + 1;}}. Algorithm \textsc{SimplifyDiff} removes
this loop and introduces the statement {\tt sum = sum +
  (N-1);}.  This helps in ${\partial \PP_N}$ having fewer and
simpler loops in a lot of cases.





\vspace{-.7ex}
\begin{lemma}
\label{lemma:simp-sound}
Program $\partial \PP'_N$ generated by \textsc{SimplifyDiff}
is such that, for all $N > 1$,
$\{\varphi(N)\} \;\PP_{N-1};{\partial \PP'_N} \; \{\psi(N)\}$ holds iff
$\{\varphi(N)\} \;\PP_{N-1};{\partial \PP_N} \; \{\psi(N)\}$ holds.
\end{lemma}




\vspace{-.2in}
\paragraph{\bfseries Generating the Difference Pre-condition $\mathbf{\partial \varphi(N)}$.}
We now present a simple syntactic algorithm, called \textsc{SyntacticDiff},
for generation of the difference pre-condition ${\partial \varphi(N)}$.
Although this suffices for all our experiments, for the sake of
completeness, we present later a more sophisticated algorithm
for generating ${\partial \varphi(N)}$ simultaneously with $\ppre(N)$.

Formally, given $\varphi(N)$, algorithm \textsc{SyntacticDiff}
generates a formula ${\partial \varphi(N)}$ such that $\varphi(N)
\rightarrow (\varphi(N-1) \wedge {\partial \varphi(N)})$. Observe that
if such a ${\partial \varphi(N)}$ exists, then $\varphi(N) \rightarrow
\varphi(N-1)$ holds as well.  Therefore, we can use the validity of
$\varphi(N) \rightarrow \varphi(N-1)$ as a test to decide the
existence of ${\partial \varphi(N)}$.
%
%

If $\varphi(N)$ is of the syntactic form $\forall i\in \{0 \ldots
N\}\; \widehat{\varphi}(i)$, then ${\partial \varphi(N)}$ is easily
seen to be $\hat{\varphi}(N)$.  
If $\varphi(N)$
is of the syntactic form $\varphi^1(N) \wedge \cdots \wedge
\varphi^k(N)$, then ${\partial \varphi(N)}$ can be computed as
${\partial \varphi^1(N)} \wedge \cdots \wedge {\partial \varphi^k(N)}$.
Finally, if $\varphi(N)$ doesn't belong to any of these syntactic forms or if
condition 2(a) of Theorem~\ref{thm:full-prog-ind-sound} is violated by
the heuristically computed ${\partial \varphi(N)}$, then we
over-approximate ${\partial \varphi_N}$ by $\true$. For a
large fraction of our benchmarks, the pre-condition $\varphi(N)$ was
$\true$, and hence ${\partial \varphi(N)}$ was also $\true$.


\vspace{-.1in}
\paragraph{\bfseries Generating the Formula $\mathbf{\ppre(N-1)}$.}
We use Dijsktra's weakest pre-condition computation to obtain
$\ppre(N-1)$ after the ``difference'' pre-condition ${\partial \varphi(N)}$ and
the ``difference'' program ${\partial \PP_N}$ have been generated.
The weakest pre-condition can always be computed using quantifier
elimination engines in state-of-the-art SMT solvers like Z3 if
${\partial \PP_N}$ is loop-free.  In such cases, we use a set of
heuristics to simplify the calculation of the weakest pre-condition
before harnessing the power of the quantifier elimination engine.
If ${\partial \PP_N}$ contains a loop, it may still be possible to obtain
the weakest pre-condition if the loop doesn't affect the post-condition.
Otherwise, we compute as much of the weakest pre-condition as can be
computed from the non-loopy parts of ${\partial \PP_N}$, and then try to
recursively solve the problem by invoking full-program induction
on ${\partial \PP_N}$ with appropriate pre- and post-conditions.

%

\begin{algorithm}[!t]
  \caption{\footnotesize \textsc{FPIVerify}({$\PP_N$: program, $\varphi(N)$}: pre-condn, {$\psi(N)$}: post-condn)}
  \label{alg:fpi}
  \scriptsize
  \begin{algorithmic}[1]
    \If{Base case check \{$\varphi(1)$\} $\PP_1$ \{$\psi(1)$\} fails}
      \State \Return ``Counterexample found!'';	
    \EndIf

    \State $\partial \varphi(N)$ := \textsc{SyntacticDiff}($\varphi(N)$);

    \State $\partial \PP_N$ := \textsc{ProgramDiff}($\PP_N$);
    \State $\partial \PP_N$ := \textsc{SimplifyDiff}($\partial \PP_N$); \Comment{Simplify and Accelerate loops}

    \State $i := 0$;
    \State $\ppre_i(N) := \psi(N)$;
    \State $c\_\ppre_i(N) := \true$;	\Comment{Cumulative conjoined pre-condition}

    \Do
      \label{line:loop-at-fpi}
      \If{ \{$c\_\ppre_i(N-1) \wedge \psi(N-1) \wedge \partial \varphi(N)$\} $\partial\PP_N$ \{$c\_\ppre_i(N) \wedge \psi(N)$\} }
      \State \Return $\true$;\label{line:loop-fpi-return-true} \Comment{Assertion verified}
      \EndIf
    \State $i := i+1$;
    \State $\ppre_i(N-1) := \textsc{LoopFreeWP}( \ppre_{i-1}(N), \partial\PP_N)$;	\Comment{Dijkstra's $\mathsf{WP}$ sans $\mathsf{WP}$-for-loops}
    \If {no new $\ppre_i(N-1)$ obtained} \Comment{Can happen if ${\partial \PP_N}$ has a loop}
    \State \Return \textsc{FPIVerify}(${\partial \PP_N}$, $c\_\ppre_i(N-1) \wedge \psi(N-1) \wedge \partial \varphi(N)$, $c\_\ppre_i(N) \wedge \psi(N)$);
    \label{line:loop-fpi-recursive}
    \Else
       \State $c\_\ppre_i(N) := c\_\ppre_{i-1}(N) \wedge \ppre_i(N)$;
    \EndIf
    \doWhile{Base case check \{$\varphi(1)$\} $\PP_1$ \{$c\_\ppre_i(1)$\} passes};
    \State \Return $\false$; \Comment{Failed to prove by full-program induction}
  \end{algorithmic}
\end{algorithm}

\vspace{-.1in}
\paragraph{\bfseries Verification by Full-program Induction.}
The basic full-program induction algorithm is presented as routine
\textsc{FPIVerify} in Algorithm~\ref{alg:fpi}.  The main steps of this
algorithm are: checking conditions 3(a), 3(b) and 3(c) of
Theorem~\ref{thm:full-prog-ind-sound} (lines 1, 18 and 10),
calculating the weakest pre-condition of the relevant part of the
post-condition (line 13), 
and strengthening the pre-condition and post-condition with the
weakest pre-condition thus calculated (line 17).  Since the weakest
pre-condition computed in every iteration of the loop ($Pre_i(N-1)$ in
line 13) is conjoined to strengthen the inductive pre-condition
($c\_Pre_i(N)$ in line 17), it suffices to compute the weakest
pre-condition of $Pre_{i-1}(N)$ (instead of $c\_Pre_i(N) \wedge
\psi(N)$) in line 13.  The possibly multiple iterations of
strengthening of pre- and post-conditions is effected by the loop in
lines 9-18.  In case the loop terminates via the {\tt return}
statement in line 11, the inductive claim has been successfully
proved. If the loop terminates by a violation of the condition in line
18, we report that verification by full-program induction failed.  In
case ${\partial \PP_N}$ has loops and no further weakest
pre-conditions can be generated, we recursively invoke
\textsc{FPIVerify} on ${\partial \PP_N}$ in line 15.  This situation
arises if, for example, we modify the example in Fig.~\ref{fig:ex}(a)
by having the statement {\tt C[t3] = N;} (instead of {\tt C[t3] = 0;})
in line 10.  In this case, ${\partial \PP_N}$ has a single loop
corresponding to the third loop in Fig.~\ref{fig:ex}(a).  The
difference program of ${\partial \PP_N}$ is, however, loop-free, and
hence the recursive invocation of full-program induction on ${\partial
  \PP_N}$ easily succeeds.


\vspace{-.1in}
\paragraph{\bfseries Generalized FPI Algorithm.}
\label{sec:fpi-ext}
While algorithm \textsc{FPIVerify} suffices for all of our
experiments, we may not always be so lucky.  Specifically, even if
${\partial \PP_N}$ is loop-free, the analysis may exit the loop in
lines 9-18 of \textsc{FPIVerify} by violating the base case check in
line 18.  To handle (at least partly) such cases, we propose the
following strategy. Whenever a (weakest) pre-condition $\ppre_i(N-1)$
is generated, instead of using it directly to strengthen the current
pre- and post-conditions, we ``decompose'' it into two formulas
$\ppre_i'(N-1)$ and ${\partial \varphi_i'(N)}$ with a two-fold intent:
(a) potentially weaken $\ppre_i(N-1)$ to $\ppre_i'(N-1)$, and (b)
potentially strengthen the difference formula ${\partial \varphi(N)}$
to ${\partial \varphi_i'(N)} \wedge {\partial \varphi(N)}$.  The
checks for these intended usages of $\ppre_i'(N-1)$ and ${\partial
  \varphi_i'(N)}$ are implemented in lines 3, 4, 5, 11 and 17 of
routine \textsc{FPIDecomposeVerify}, shown as
Algorithm~\ref{alg:fpi-ext}.  This routine is meant to be invoked as
\textsc{FPIDecomposeVerify}$(i)$ after each iteration of the loop in
lines 9-18 of routine \textsc{FPIVerify} (so that $\ppre_i(N)$,
$c\_\ppre_i(N)$ etc. are initialized properly).  In general, several
``decompositions'' of $\ppre_i(N)$ may be possible, and some of them
may work better than others.  \textsc{FPIDecompseVerify} permits
multiple decompositions to be tried through the use of the
$\textsc{NextDecomposition}$ and $\textsc{HasNextDecomposition}$
functions.  Lines 22-25 of \textsc{FPIDecomposeVerify} implement a
simple back-tracking strategy, allowing a search of the space of
decompositions of $\ppre_i(N-1)$.  Observe that when we use
\textsc{FPIDecomposeVerify}, we simultaneously compute a difference
formula $({\partial \varphi'_i(N)} \wedge {\partial \varphi(N)})$ and
an inductive pre-condition $(c\_\ppre_{i-1}(N) \wedge \ppre_i'(N))$.

\begin{algorithm}[!t]
  \caption{\footnotesize \textsc{FPIDecomposeVerify}( i : integer )}
  \label{alg:fpi-ext}
  \scriptsize
  \begin{algorithmic}[1]
    \Do
      \State $\langle\ppre_i'(N-1), \partial \varphi_i'(N)\rangle$ := $\textsc{NextDecomposition}(\ppre_i(N-1))$;
      \State Check if (a) $\partial \varphi_i'(N) \wedge \ppre_i'(N-1)  \rightarrow  \ppre_i(N-1)$,\\ 
      \hspace*{0.6in}(b) $\varphi(N) \rightarrow \varphi(N-1) \wedge \left(\partial \varphi'_i(N) \wedge \partial \varphi(N)\right)$, \\
      \hspace*{0.6in}(c) $\PP_{N-1}$ does not update any variable or array element in $\partial \varphi_i'(N)$
      \If{any check in lines 3-5 fails}
      \If{$\textsc{HasNextDecomposition}(\ppre_i(N-1))$}
      \State \textbf{continue};
      \Else
      \State \Return $\false$;
      \EndIf
      \EndIf
      
      \If{\{$c\_\ppre_{i-1}(N-1) \wedge \psi(N-1) \wedge \ppre_i(N-1) \wedge \partial \varphi(N)$\} $\partial\PP_N$ \{$c\_\ppre_{i-1}(N) \wedge \psi(N) \wedge \ppre_i'(N)$\}}
        \State \Return $\true$; \Comment{Assertion verified}
      \Else
        \State $c\_\ppre_i(N) := c\_\ppre_{i-1}(N) \wedge \ppre_i'(N)$;
        \State $i := i+1$;
        \State $\ppre_i(N-1) := \textsc{LoopFreeWP}( \ppre_{i-1}'(N), \partial\PP_N)$; \Comment{Dijkstra's $\mathsf{WP}$ sans $\mathsf{WP}$-for-loops}

        \If {\{$\varphi(1)$\} $\PP_1$ \{$c\_\ppre_{i-1}(1) \wedge \ppre_i(1)$\}  does not hold}
           \State $i := i-1$;
        \Else
           \State $prev\_\partial \varphi(N)$ := $\partial \varphi(N)$;
           \State $\partial \varphi(N)$ := $\partial \varphi'_{i-1}(N) \wedge \partial \varphi(N)$;
           \If{\textsc{FPIDecomposeVerify}$(i)$ returns $\false$}
             \State $i := i-1$;   $\partial \varphi(N)$ := $prev\_\partial \varphi(N)$;
           \Else
             \State \Return $\true$;
          \EndIf
        \EndIf
      \EndIf
      \doWhile{$\textsc{HasNextDecomposition}(\ppre_i(N-1))$};

      \State \Return $\false$;
    \end{algorithmic}
\end{algorithm}





\vspace{-.7ex}
\begin{lemma}
\label{lemma:fpi-alg}
Algorithms \textsc{FPIVerify} and \textsc{FPIDecomposeVerify} ensure
conditions 2 and 3 of Theorem \ref{thm:full-prog-ind-sound} upon
successful termination.
\end{lemma}
\vspace{-.7ex}

%
\noindent
While we have presented our technique focusing on a single symbolic
parameter $N$, a straightforward extension works for multiple
independent parameters, multiple independent array sizes, different
induction directions, and non-uniform loop termination conditions.%

\vspace{-.1in}

\paragraph{\bfseries Limitations.}
There are several scenarios under which full-program induction may not
produce a conclusive result. Currently, we only analyze programs with
non-nested loops with $+, -, \times, \div$ expressions in assignments.
We also do not handle branch conditions that are dependent on the
parameter N (this doesn't include loop conditions, which are handled
by unrolling the loop).  The technique also remains inconclusive when
the difference program ${\partial \PP_N}$ does not have fewer loops
than the original program. Reduction in verification complexity of the
program, in terms of the number of loops and assignment statements
dependent on $N$, is crucial to the success of full-program
induction. Finally, our technique may fail to verify a correct program
if the heuristics used for weakest pre-condition either fail or return
a pre-condition that causes violation of the base case check in line
18 of \textsc{FPIVerify}.  Despite these limitations, our experiments
show that full-program induction performs remarkably well on a large
suite of benchmarks.






\vspace{-3ex}
\section{Implementation and Experiments}
\vspace{-1ex}
\label{sec:experiments}
We have implemented our technique in a prototype tool called {\ourtool},
available at~\cite{vajra-artifact}.
It takes a C program in SVCOMP format as input.
The tool, written in {\tt C++}, is built on top of the LLVM/CLANG~\cite{clang} $6.0.0$
compiler infrastructure and uses {\zthree}~\cite{z3} v$4.8.7$ as the SMT solver
to prove Hoare triples for loop-free programs.

We have evaluated {\ourtool} on a test-suite of $42$ safe benchmarks inspired
from different algebraic functions that compute polynomials as well as a
standard array operations such as copy, min, max and compare.
Our programs take a symbolic parameter $N$ which specifies the size of each array
as well as the number of times each loop executes. 
%
Assertions, possibly quantified, are (in-)equalities over array
elements, scalars and (non-)linear polynomial terms over $N$.

All experiments were performed on a Ubuntu 18.04 machine with 16GB RAM and
running at 2.5 GHz. 
We have compared {\ourtool} against
{\viap}(v1.0)~\cite{viap}, {\veriabs}(v1.3.10)~\cite{veriabs},
{\booster}~(v0.2)\cite{booster}, {\vaphor}(v1.2)~\cite{vaphor}
and {\freqhorn}(v3)~\cite{madhukar19}.
C programs were manually converted to mini-Java as required by {\vaphor}
and CHC's as required by {\freqhorn}.
Our results are shown in Table \ref{tab:exp-results}.
{\ourtool} verified $36$ benchmarks, compared to $23$ verified by {\viap},
$12$ by {\veriabs}, $8$ by {\booster}, $5$ each by {\vaphor} and {\freqhorn}.
{\ourtool} was unable to compute the difference program for $5$ benchmarks
and was inconclusive on $1$ benchmark.

\begin{table*}[!t]
\tiny
\begin{minipage}[b]{0.45\hsize}\centering
\begin{tabular}{|c|c|c|c|c|c|c|c|}
\hline
\textsc{Name} & \#L & T1 & T2 & T3 & T4 & T5 & T6 \\ \hline \hline
pcomp     & 3 & \cmark 0.68 & TO & TO & \textbf{?}0.23 & TO & \textbf{?}0.58 \\ \hline
ncomp     & 3 & \cmark 0.68 & TO & TO & \textbf{?}0.41 & TO & \textbf{?}0.68 \\ \hline
eqnm2     & 2 & \cmark 0.52 & TO & TO & \textbf{?}0.07 & TO & \textbf{?}0.59 \\ \hline
eqnm3     & 2 & \cmark 0.53 & TO & TO & \textbf{?}0.07 & TO & \textbf{?}0.56 \\ \hline
eqnm4     & 2 & \cmark 0.51 & TO & TO & \textbf{?}0.07 & TO & \textbf{?}0.60 \\ \hline
eqnm5     & 2 & \cmark 0.55 & TO & TO & \textbf{?}0.07 & TO & \textbf{?}0.58 \\ \hline
sqm  & 2 & \cmark 0.51 & \cmark 69.7 & TO & \textbf{?}0.11 & TO & \textbf{?}0.57 \\ \hline
res1     & 4 & \cmark 0.17 & TO & TO & TO & TO & TO \\ \hline
res1o  & 4 & \cmark 0.18 & TO & TO & TO & TO & TO \\ \hline
res2     & 6 & \cmark 0.20 & TO & TO & TO & TO & TO \\ \hline
res2o  & 6 & \cmark 0.22 & TO & TO & TO & TO & TO \\ \hline
ss1 & 4 & \cmark 0.40 & TO & TO & \xmark 0.13 & \textbf{?}19.2 & \textbf{?}1.7 \\ \hline
ss2 & 6 & \cmark 0.46 & TO & TO & \xmark 0.13 & TO & \textbf{?}9.7 \\ \hline
ss3 & 5 & \cmark 0.35 & TO & TO & \xmark 0.13 & TO & \textbf{?}2.1 \\ \hline
ss4 & 4 & \cmark 0.29 & TO & TO & \xmark 0.13 & TO & \textbf{?}1.6 \\ \hline
ssina & 5 & \cmark 0.41 & \cmark 72.5 & TO & TO & TO & \textbf{?}2.0 \\ \hline
sina1  & 2 & \cmark 0.56 & \cmark 65.4 & TO & TO & TO & TO \\ \hline
sina2  & 3 & \cmark 0.69 & \cmark 66.5 & TO & TO & TO & TO \\ \hline
sina3  & 4 & \cmark 0.83 & TO & TO & TO & TO & TO \\ \hline
sina4  & 4 & \cmark 0.85 & TO & TO & TO & TO & TO \\ \hline
sina5  & 5 & \cmark 0.93 & TO & TO & TO & TO & TO \\ \hline
\end{tabular}
\end{minipage}
\begin{minipage}[b]{0.52\hsize}\centering
\begin{tabular}{|c|c|c|c|c|c|c|c|}
\hline
\textsc{Name} & \#L & T1 & T2 & T3 & T4 & T5 & T6 \\ \hline \hline
zerosum1 & 2   & \cmark 0.33 & \cmark 62.0 & \cmark 11 & \cmark 0.77       & \xmark 0.29 & TO \\ \hline
zerosum2 & 4   & \cmark 0.46 & \cmark 75.8 & \cmark 18 & TO                        & \xmark 1.64 & TO \\ \hline
zerosum3 & 6   & \cmark 0.59 & \cmark 73.1 & \cmark 39 & TO                        & \xmark 3.13 & TO \\ \hline
zerosum4 & 8   & \cmark 0.76 & \cmark 76.1 & TO & \textbf{?}18.2 & \xmark 6.85 & TO \\ \hline
zerosum5 & 10 & \cmark 0.97 & \cmark 80.6 & TO & \textbf{?}16.5 & \xmark 10.4 & TO \\ \hline
zerosumm2 & 4  & \cmark 0.46 & \cmark 71.5 & \cmark 24 & TO                       & \xmark 1.22 & TO \\ \hline
zerosumm3 & 6  & \cmark 0.59 & \cmark 70.9 & TO & TO                      & \xmark 5.22 & TO \\ \hline
zerosumm4 & 8  & \cmark 0.77 & \cmark 76.4 & TO & \textbf{?}16.7  & \xmark 12.39 & TO \\ \hline
zerosumm5 & 10 & \cmark 0.98 & \cmark 81.7 & TO & \textbf{?}18.7 & \xmark 22.8 & TO \\ \hline
zerosumm6 & 12 & \cmark 1.29 & \cmark 86.8 & TO & \textbf{?}16.1 & TO & TO \\ \hline
copy9 & 9 &  \cmark 0.69 & \cmark 86.8 & \cmark 3.91 & \cmark 18.8 & TO & \cmark 0.67 \\ \hline
min & 1 &  \cmark 0.48 & \cmark 23.6 & \cmark 3.82 & \cmark 0.52 & \cmark 0.14 & \cmark 0.13 \\ \hline
max & 1 &  \cmark 0.46 & \cmark 25.4 & \cmark 4.70 & \cmark 1.0 & \cmark 0.28 & \cmark 0.18 \\ \hline
compare & 1 &  \cmark 0.82 & \cmark 18.8 & \cmark 17.9 & \cmark 0.06 & \cmark 0.84 & \cmark 0.31 \\ \hline
conda     & 3 &  \cmark 0.72 & \cmark 13.9 & TO & \cmark 0.07 & \cmark 0.09 & TO \\ \hline
condn     & 1 &  \textbf{?}0.51 & \cmark 14.7 & \cmark 18.9 & \cmark 0.02 & \cmark 0.15 & \cmark 0.20 \\ \hline
condm     & 2 & \textbf{?}0.59 & \cmark 20.5 & \cmark 16.7 & \cmark 0.04 & TO & \textbf{-}  \\ \hline
condg & 3 & \textbf{?}0.52 & TO & TO & TO & TO & TO \\ \hline
modn     & 2 & \textbf{?}0.63 & \cmark 22.6 & TO & \textbf{-} & TO & TO \\ \hline
mods  & 4 & \textbf{?}0.61 & TO & \cmark 18.2 & $ \textbf{-} $ & \textbf{-} & \textbf{-} \\ \hline
modp  & 2 & \textbf{?}0.71 & \cmark 17.3 & \cmark 40 & \textbf{-} & \textbf{?}32 & \textbf{-} \\ \hline
\end{tabular}
\end{minipage}
\vspace{3ex}
\caption{
  First column is the benchmark name.
  Second column indicates the number loops in the benchmark
  (excluding the assertion loop). Successive columns indicate the results generated
  by tools and the time taken where
  T1 is \ourtool,
  T2 is \viap,
  T3 is \veriabs,
  T4 is \booster,
  T5 is \vaphor,
  T6 is \freqhorn.
  \cmark indicates assertion safety,
  \xmark indicates assertion violation,
  \textbf{?} indicates unknown result, and
  \textbf{-} indicates an abrupt stop.
  All the times are in seconds.
  TO is time-out of 100 secs.
 }
\label{tab:exp-results}
\vspace{-4ex}
\end{table*}




{\ourtool} verified $17$ benchmarks on which {\viap} diverged,
primarily due to the inability of {\viap}'s heuristics to get closed
form expressions.  {\viap} verified $4$ benchmarks that could not be
verified by the current version of {\ourtool} due to syntactic
limiations.  {\ourtool}, however, is two orders of magnitude faster
than {\viap} on programs that were verified by both.  {\ourtool}
proved $28$ benchmarks on which {\veriabs} diverged.  {\veriabs} ran
out of time on programs where loop shrinking and merging abstractions
were not strong enough to prove the assertions.  {\veriabs} reported 1
program as unsafe due to the imprecision of its abstractions and it
proved $4$ benchmarks that {\ourtool} could not.
{\ourtool} verified $30$ benchmarks that {\booster} could not.
{\booster} reported $4$ benchmarks as unsafe due to imprecise
abstractions, its fixed-point computation
engine reported unknown result on $12$ benchmarks and it ended
abruptly on $3$ benchmarks. 
{\booster} also proved $2$ benchmarks that couldn't be handled by the
current version of {\ourtool} due to syntactic limitations.
{\ourtool} verified $32$ benchmarks on which {\vaphor} was inconclusive.
Distinguished cell abstraction in {\vaphor}  is unable to prove safety of programs,
when the value at each array index needs to be tracked.
{\vaphor} reported $9$ programs unsafe due to imprecise abstraction,
returned unknown on 2 programs and ended abruptly on $1$ program.
{\vaphor} proved a benchmark that {\ourtool} could not.
%
{\ourtool} verified $32$ programs on which {\freqhorn} diverged,
especially when constants and terms that appear in the inductive invariant
are not syntactically present in the program.
{\freqhorn} ran out of time on $22$ programs,
reported unknown result on $12$ and
ended abruptly on $3$ benchmarks.
{\freqhorn} verified a benchmark with a single loop that {\ourtool} could
not.
%
On an extended set of $231$ benchmarks,
{\ourtool} verified $110$ programs out of $121$ safe programs,
falsified $108$ out of $110$ unsafe programs, and was inconclusive
on the remaining $13$ programs.





\vspace{-2ex}
\section{Conclusion}
\label{sec:conc}
\vspace{-2ex}
We presented a novel property-driven verification method that performs
induction over the entire program via parameter $N$.  Significantly,
this obviates the need for loop-specific invariants.
Experiments show that full-program induction performs remarkably well
vis-a-vis state-of-the-art tools for analyzing array manipulating
programs.  Further improvements in the algorithms for computing
difference programs and for strengthening of pre- and post-conditions
are envisaged as part of future work.



\section*{Data Availability Statement}
\vspace{-2ex}
The datasets generated and analyzed during the current study are available in the figshare repository: https://doi.org/10.6084/m9.figshare.11875428.v1

\bibliographystyle{splncs04}
\bibliography{fpi}


\vfill

{\small\medskip\noindent{\bf Open Access} This chapter is licensed under the terms of the Creative Commons\break Attribution 4.0 International License (\url{http://creativecommons.org/licenses/by/4.0/}), which permits use, sharing, adaptation, distribution and reproduction in any medium or format, as long as you give appropriate credit to the original author(s) and the source, provide a link to the Creative Commons license and indicate if changes were made.}

{\small \spaceskip .28em plus .1em minus .1em The images or other third party material in this chapter are included in the chapter's Creative Commons license, unless indicated otherwise in a credit line to the material.~If material is not included in the chapter's Creative Commons license and your intended\break use is not permitted by statutory regulation or exceeds the permitted use, you will need to obtain permission directly from the copyright holder.}

\medskip\noindent\includegraphics{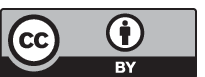}


\end{document}